\documentclass[11pt,english]{article}

\usepackage{latexsym}
\usepackage{epsfig}
\usepackage{amssymb}
\usepackage{array}
\usepackage{float}
\usepackage{color}
\graphicspath{{figures/}{./}}

\newenvironment{proof}{{\bf Proof. } }{{\hfill $\Box$}\vspace{.5pc}}
\newtheorem{theorem}{Theorem}[section]
\newtheorem{corollary}[theorem]{Corollary}
\newtheorem{definition}[theorem]{Definition}
\newtheorem{lemma}[theorem]{Lemma}

\newtheorem{remark}[theorem]{Remark}

\floatstyle{ruled}

\floatstyle{ruled}
\newfloat{algo}{thp}{lop}
\floatname{algo}{Algorithm}

\newcommand{\tO}{\mathtt{O}}

\newcommand{\A}[1]	{\mathtt{Async}_{{#1}}}

\newcommand{\BEGLIST}{\begin{list}{}{\partopsep -3pt \parsep -2pt}}
\newcommand{\ENDLIST}{\end{list}}
\newcommand{\ie}{\emph{i.e., }}
\newcommand{\eg}{\emph{e.g., }}

\begin{document}

\title{Deaf, Dumb, and Chatting Robots\\
        \small{Enabling Distributed Computation and Fault-Tolerance Among Stigmergic Robot}\footnote{The work of
Shlomi Dolev and Michael Segal has been partially supported by US Air
Force grant.}
}
\author{Yoann Dieudonn\'e$^\dagger$ \and Shlomi Dolev$^\ddagger$ \and Franck Petit$^\ast$ \and Michael Segal$^\star$\\
 $^\dagger$  MIS Lab., University of Picardie Jules Verne, France\\
 $^\ddagger$ Department of Computer Science, Ben-Gurion University of the Negev, Israel\\
 $^\ast$     INRIA, LIP Lab., University of Lyon, France\\
 $^\star$    Communication Systems Engineering Dept, Ben-Gurion University of the Negev, Israel
}
\date{}
\maketitle

\begin{abstract}
We investigate ways for the exchange of information (explicit communication) among deaf and dumb mobile 
robots scattered in the plane.
We introduce the use of movement-signals (analogously to flight signals and bees waggle) 
as a mean to transfer messages, enabling the use of distributed algorithms among the robots. 
We propose one-to-one deterministic movement protocols that implement explicit communication.

We first present protocols for synchronous robots. We begin with a
very simple coding protocol for two robots.  Based on on this
protocol, we provide one-to-one communication for any system of $n
\geq 2$ robots equipped with observable IDs that agree on a common
direction (sense of direction). We then propose two solutions
enabling one-to-one communication among anonymous robots. Since the
robots are devoid of observable IDs, both protocols build
recognition mechanisms using the (weak) capabilities offered to the
robots.  The first protocol assumes that the robots agree on a
common direction and a common handedness (chirality), while the
second protocol assumes chirality only.
Next, we show how the movements of robots can provide implicit
acknowledgments in asynchronous systems. We use this result to
design asynchronous one-to-one communication with two robots only.
Finally, we combine this solution with the schemes developed in
synchronous settings to fit the general case of asynchronous
one-to-one communication among any number of robots. 

Our protocols enable the use of distributing
algorithms based on message exchanges among swarms of Stigmergic robots.  
Furthermore, they provides robots equipped with means of communication to 
overcome faults of their communication device. 
\smallskip

\end{abstract}

%
%

\section{Introduction}
Research for achieving coordination among teams (or, swarms) of mobile robots
is challenging and has a great scientific and practical implications.
Swarms of mobile robots are used, and are planned to be even used more, in several critical situations.
Swarms provide the ability to measure properties, collect information and
act in a given (sometimes dangerous) physical environment.
Numerous of potential applications exist for such multi-robot systems, 
to name only a very few: environmental monitoring, large-scale construction,
risky area surrounding or surveillance, and exploration of awkward environments.

In a given environment, the ability for the swarm of robots to succeed in
the accomplishment of the assigned task greatly depends on the capabilities
the robots have, that is, their moving capacities and sensory organs.
Existing technologies allow to consider robots equipped with sensor devices
for vision (camera, radar, sonar, laser, etc.) and means of communication (wireless devices).
Means of communication are even required to enable classical distributed algorithms and to
achieve completion of some tasks like information relay, surveillance, or intelligence activities.

An interesting question is ``{\em What happens if the means of
communication are lost or not exist?}''  In that case, the robots
can observe the location of the robots but cannot communicate with
them. Such robots are called {\em deaf and dumb}. There are numerous
realistic scenarios where there is no mean to communicate. Such
scenarios are easy to figure out, \eg
\begin{itemize}
\item wireless devices are faulty;
\item the robots evolve in zones with blocked wireless communication, \eg
 hostile environments where communication are scrambled or forbidden;
\item physical constraints prevent to fit out any wireless device on the robots.
\end{itemize}
The latter can arise for instance when no space is available on the robots or
the robots are too small with respect to the size of the wireless device. This
would be the case with swarms of nano-robots.

As a matter of fact, the question of solving distributed tasks with
swarms of deaf and dumb robots is not a novel one. This question has
been extensively carried in different fields of computed science
such as artificial intelligence~\cite{M95}, control
theory~\cite{FN88,N92,K97}, and recently in the distributed
computing field~\cite{SY99,P02}. Some of these approaches are
inspired of biology studies of the animal behavior, mainly the
behavior of social insects~\cite{BHD94}. Indeed, these social
systems present an intelligent collective behavior although they are
composed by simple individuals with extremely limited capabilities.
Solutions to problems ``naturally'' emerge from the
self-organization and indirect communication of these individuals.
The capacity to communicate using such {\em indirect communication}
(or, {\em implicit} communication) is referred to as {\em stigmergy}
in the biological literature~\cite{G59}. There are numerous examples
of such indirect communication in nature, for instance from ants and
termites communicating using pheromones or from bees, communicating
performing waggle dances to find shortest paths between their nest
and food sources. The question of whether the waggle of bees is a
language or not is even an issue~\cite{WW90}.

However, stigmergy allows to make a given task only.  Communications
are not consider as a task by itself. In other words, the stigmergic
phenomenon provides indirect communication, a guidance for a
specific work. Even if, sometime stigmergy allows insects to modify
their physical environment---this phenomenon is sometime referred to
as \emph{sematectonic stigmergy}~\cite{W75}--- stigmergy never
provides a communication task by itself. In other words, stigmergy
does not allow tasks as chatting, intelligence activities, or
sending information which are not related with a given task.

In this paper, we investigate ways for the exchange of information
among deaf and dumb mobile robots scattered in the plane. In the
sequel, we refer to this task as {\em explicit
communication}---sometime, also referred to as {\em direct}
communication \cite{M95}. Explicit communication enables the use of
distributed algorithms among the robots. We study the possibility to
solve this problem deterministically when robots are communicating
only by moving with respect to certain capabilities, namely
synchrony, identities, sense of orientation, and chirality.

\paragraph{Contribution.}
We introduce the use of movement-signals (analogously to flight
signals and bees waggle) as a mean to transfer messages between deaf
and dumb robots. We propose six one-to-one deterministic protocols
that implement explicit communication.

We first present four protocols for synchronous robots. We begin
with a very simple protocol for two robots, showing how high-level
information can be coded by simple moves. Based on this protocol, we
provide one-to-one communication into any system of $n \geq 2$
robots equipped with observable IDs that agree on a common direction
(sense of direction). We then propose two solutions enabling
one-to-one communication among anonymous robots. Since the robots
are devoid of observable IDs, both protocols build recognition
mechanisms using the (weak) capabilities offered to the robots.  The
third algorithm assumes that the robots have sense of direction and
agree on a common handedness (chirality).  The fourth protocol
assumes chirality only.

In synchronous systems, since every robot is active at each time
instant, there is no concern with the receipt of the messages.
Indeed, every movement is seen by every robot.  As a consequence, no
message acknowledgment is required. By contrast, in asynchronous
settings, in each computation steps, some robots can be inactive.
Thus, few robot movements can be missed out, and as a consequence,
some messages. We show how the robot movements can provide implicit
acknowledgments in asynchronous settings. We use this result in the
design of asynchronous one-to-one communication with two robots
only.  Finally, we combine this solution with the schemes developed
in synchronous settings to fit the general case of asynchronous
one-to-one communication among any number of robots.

Note that our protocols---either synchronous or asynchronous---can be 
easily adapted to implement efficiently one-to-many or one-to-all 
explicit communication. 
Also, in the context of robots (explicitely) communicating by means 
of communication (\eg wireless), since our protocols allow robots to 
explicitly communicate even if their communication devices are faulty, 
in a very real sense, our solution can serve as a communication backup,
\ie it provides {\em fault-tolerance} 
by allowing the robots to communicate without means of communication (wireless device).

%

\paragraph{Related Work.}
The issue of handling swarms of robots using deterministic
distributed algorithm was first studied in~\cite{SY96,SY99}. Beyond
supplying formal correctness proofs, the main motivation is to
understand the relationship between the capabilities of the robots
and the solvability of given tasks.  For instance, ``{\em Assuming
that the robots agree on a common direction (having a compass), which
tasks are they able to deterministically achieve}?'', or ``{\em What
are the minimal condition to elect a leader deterministically}?''

As a matter of fact, the motivation turns out to be the study of the
minimal level of ability the robots are required to have in the
accomplishment of some basic cooperative tasks in a deterministic
way.  Examples of such tasks are specific problems, so that
\emph{pattern formation}, \emph{line formation}, \emph{gathering},
\emph{spreading}, and \emph{circle formation}---refer for instance
to~\cite{SY99,FPSW99,DK02,P02,CFPS03,PS06,DLP08,CP08} for this
problems,--- or more ``classical'' problems in the field of
distributed systems, such that {\em leader election}
\cite{P02,FPSW99,FPSW01,DP07c}. To the best of our knowledge, no
previous work addresses the problem of enabling explicit
communication in swarms of robots.

\paragraph{Roadmap.}
In the next section, we describe the model and the problem considered in this paper.
The two following sections (Section~\ref{sec:sync} and~\ref{sec:async}) are devoted to 
synchronous and asynchronous settings, respectively.  
(Due to the lack of space, proofs are given in the appendix.)
Finally, we make some concluding remarks in Section~\ref{sec:conclu}.  In the same 
section, extensions and open problems are also discussed. 

\section{Preliminaries.}
\label{sec:model}

In this section, we first define the distributed system considered
in this paper. We then state the problem to be solved.

\paragraph{Model.}
We adopt the model introduced in~\cite{SY96}, below referred to as
\emph{Semi-Synchronous Model (SSM)}. The distributed system
considered in this paper consists of $n$ mobile \emph{robots}
(\emph{agents} or \emph{sensors}).
Any robot can observe, compute and move with an infinite decimal precision.
The robots are equipped with sensors enabling to detect the instantaneous position of the other robots in the plane.
Viewed as points in the Euclidean plane, the robots are mobile and autonomous.
There is no kind of explicit communication medium.

Each robot $r$ has its own local $x$-$y$ Cartesian coordinate system with its own unit measure.
Given an $x$-$y$ Cartesian coordinate system, the \emph{handedness} is the way in which
the orientation of the $y$ axis (respectively, the $x$ axis) is inferred according to
the orientation of the $x$ axis (resp., the $y$ axis).
The robots are assumed to have the ability of \emph{chirality}, \ie
the $n$ robots share the same handedness.
%
We consider \emph{non-oblivious} robots, \ie every robot can remember its previous observations, computations,
or motions made in any previous step.

We assume that the system is either \emph{identified} or
\emph{anonymous}. In the former case, each robot $r$ is assumed to
have a visible (or, observable) identifier denoted $id_r$ such that,
for every pair $r,r'$ of distinct robots, $id_r \ne id_{r'}$. In the
latter, no robot is assumed to have a visible identified. In this
paper, we will also discuss whether the robots agree on the
orientation of their $y$-axis or not. In the former case, the robots
are said to have the \emph{sense of direction}. (Note that since the
robots have the ability of chirality, when the robots have the sense
of direction, they also agree on their $x$-axis).

Time is represented as an infinite sequence of time instants $t_0, t_1, \ldots, t_j, \ldots$
Let $P(t_j)$ be the set of the positions in the plane occupied by the $n$
robots at time $t_j$ ($j\geq0$). For every $t_j$, $P(t_j)$ is called the \emph{configuration}
of the distributed system in $t_j$.
$P(t_j)$ expressed in the local coordinate system of any robot $r_i$ is called a \emph{view}.
At each time instant $t_j$ ($j\geq 0$), each robot $r_i$ is either {\it active} or {\it inactive}.
The former means that, during the computation step $(t_j,t_{j+1})$, using
a given algorithm, $r_i$ computes in its local coordinate system a position $p_i(t_{j+1})$ depending
only on the system configuration at $t_j$, and moves towards $p_i(t_{j+1})$---$p_i(t_{j+1})$ can be equal to
$p_i(t_j)$, making the location of $r_i$ unchanged.
In the latter case, $r_i$ does not perform any local computation and remains at the same position.

The concurrent activation of robots is modeled by
the interleaving model in which the robot activations are driven by a \emph{uniform fair scheduler}.
In this paper, we discuss whether the system is \emph{synchronous} or \emph{asynchronous}.  In the former
case, every robot is active at each instant.  The latter means that at least one robot is required to
be active at each instant.

In every single activation, the distance traveled by any robot $r$
is bounded by $\sigma_r$. So, if the destination point computed by
$r$ is farther than $\sigma_r$, then $r$ moves toward a point of at
most $\sigma_r$. This distance may be different between two robots.

\paragraph{Problem.}
{\em Indirect} communication is the result of the observations of other robots.
Using indirect communication, we aim to implement {\em Direct} communication, that is a purely
communicative act, one with the sole purpose of transmitting messages~\cite{M95}.
In this paper, we consider directed communication that aim at a particular receiver.
Such communication are said to be {\em one-to-one}, specified as follows:
($Emission$) If a robot $r$ wants to send a message $m$ to a robot $r'$, then
$r$ eventually sends $m$ to $r'$; ($Receipt$) Every robot eventually receives every message which is
meant to it.

Note that the above specification induces that $r$ is able to
address $r'$.  This implies that any protocol solving the above
specification have to come up with ($1$) {\em Routing} mechanism and
($2$) {\em Naming} mechanism, the latter, in this context of
anonymous robots.

The specification also induces that the robots are able to communicate explicit messages.
So, any one-to-one communication protocol in our model
has to be able ($3$) to code explicit messages with implicit communication, 
\ie with (non-ambiguous) movements.

\section{One-to-One Communications in Synchronous Settings}
\label{sec:sync}

It is well-known that in synchronous systems, Emission and Receipt
properties are easily verified. Thus, we focus on the three
additional properties described above. We first present a solution
working for two robots only that shows how explicit information can
be easily coded making moves. The second protocol addresses Routing
property. Combined with the simple coding mechanism, it provides
one-to-one communication into any system of $n \geq 2$ robots
equipped with observable IDs and sense of direction. Our third and
fourth solutions deal with anonymous networks.  Since the robots are
devoid of observable IDs, both protocols provide Naming mechanism,
\ie they build recognition mechanisms using the capabilities offered
to the robots.  The third algorithm assumes that the robots have
sense of direction. The last one assumes robots with chirality only
(no common sense of direction).

\subsection{Coding With Two Robots}

Each even step $t_i \mapsto t_{i+1}$ ($i=2p$, $p \geq 0$) is used by each robot
to send a bit in $\{0,1\}$.  To send ``$0$'' (``$1$'', respectively) to the other robot $r'$,
a robot, $r$, moves on its right (left, resp.) with respect to the direction given by $r'$---refer to Figure\ref{com1}.
To avoid that the robots do not go either too far or near of each other,
each odd step $t_{i+1} \mapsto t_{i+2}$ is used by the robots to come back to its first position.

\begin{figure}[!htbp]
\begin{center}
    \epsfig{file=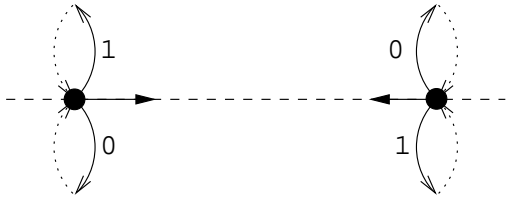, width=0.4\linewidth}
\end{center}
\caption{One-to-one communication for $2$ synchronous robots.}
\label{com1}
\end{figure}

Note that this simple protocol is \emph{silent}~\cite{DGS96} in the following sense:
a robot eventually moves if it has some messages to transmit.
Also, note that if each robot $r$ knows the maximum distance $\sigma_{r'}$ that the other robot $r'$ can cover in one step,
then, the protocol can easily be adapted to reduce the number of moves made by the robots to send bytes.
In that case, the total distance $2\sigma_{r'}$ made by $r'$ on its right and its left can be divided by the number of possible
bytes send by the robots. Then, $r'$ moves on its right or on its left of a distance corresponding to the byte sent.

\subsection{Routing With Identified Robots Having Sense Of Direction}

First, for each robot being \emph{a priori} surrounded of several
robots, our method requires to include a mechanism avoiding
collisions.  Next, it must include a technique allowing any robot to
send messages to a specific robot. In order to deal with the
collision avoidance, we use the following concept, \emph{Voronoi
diagram}, in the design of our method.

\begin{definition}[Voronoi diagram]\cite{A91}
\label{definition1}
The Voronoi diagram of a set of points $P=\{p_1,p_2,\cdots,p_n\}$ is a subdivision of the plane into $n$ cells, one for each point in $P$.
The cells have the property that a point $q$ belongs to the Voronoi cell of point $p_i$ iff for any other point $p_j \in P$,
$dist(q,p_i)< dist(q,p_j)$ where $dist(p,q)$ is the Euclidean distance between $p$ and $q$. In particular, the strict inequality
means that points located on the boundary of the Voronoi diagram do not belong to any Voronoi cell.
\end{definition}

Our protocol requires the two following preprocessing steps---executed at time $t_0$:
\begin{enumerate}
\item Each robot computes the Voronoi Diagram, each Voronoi cell being centered on a
robot position---refer to Case~($a$) in Figure~\ref{fig:ID}, the plain line.
Every robot is allowed to move into its Voronoi cell only.  This ensures the collision avoidance.

\item For each associated Voronoi cell $c_r$ of robot $r$, each robot $r$ computes the corresponding \emph{granular} $g_r$,
the largest disc of radius $R_r$ centered on $r$ and enclosed in
$c_r$---Case~($a$) in Figure~\ref{fig:ID}, the dotted lines. Notice
that the radii of different disks might be different. Each granular
is sliced into $2n$ slices, \ie the angle between two adjacent
diameters is equal to $\frac{\pi}{n}$. Each diameter is labeled from
$0$ to $n-1$, the diameter labeled by $0$ being aligned on the
North, the other are numbered in the natural order following the
clockwise direction.
\end{enumerate}

Note that since the robots share a common handedness (chirality),  
they all agree on the same clockwise direction. Having a common sense
of direction, they all agree on the same granular and slice numbering.
Once these two preprocessing steps are done, the protocol follows a similar
scheme as the case with two robots: when a robot $r$ wants to send a
bit to any robot $r'$, to send $0$ (resp., $1$), $r$ moves inside
$g_r$ on the Northern/Eastern/North-Eastern (resp.,
Southern/Western/South-Western) side on the diameter labeled $r'$.
Next, it comes back to its first position, \ie the center of $g_r$.

\begin{figure}[!htbp]
\begin{center}
  \begin{minipage}[t]{0.4\linewidth}
    \centering
    \epsfig{file=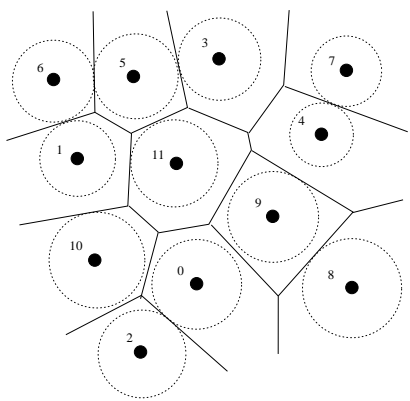, width=.8\linewidth}\newline
    {\footnotesize ($a$) An example showing the system with $12$ robots,
       viewed by each robot after the two preprocessing phases.}
  \end{minipage}
  \begin{minipage}[t]{0.5\linewidth}
    \centering
    \epsfig{file=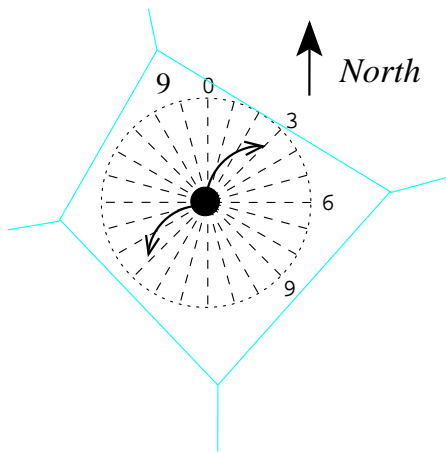, width=.7\linewidth}\newline
    {\footnotesize ($b$) Robot $9$ sends either ``$0$'' or ``$1$'' to Robot $3$.}
  \end{minipage}
 \caption{One-to-one communication with $n$ identified robots having sense of direction.}
\label{fig:ID}
\end{center}
\end{figure}



\subsection{Naming With Anonymous Robots Having Sense Of Direction}

It may seem difficult to send a message to a specific robot since no robot has an observable ID. However,
it is shown in~\cite{FPSW99} that if the robots have sense of direction and chirality, then they can agree on
a total order over the robots.  This is simply obtained as follows:  Each robot $r$ labels every observed robot
with its local $x-y$ coordinate in the local coordinate system of $r$.  Even if the
robots do not agree on their metric system, by sharing the same $x$- and $y$-axes, they agree on the same
order.

\subsection{Naming With Anonymous Robots Without Sense Of Direction}

By contrast to the previous case, with the lack of sense of
direction, the robots cannot deterministically agree on a common
labeling of the cohort.  For instance, in Figure~\ref{fig:noSoD},
the symmetry of the configuration prevent the robots to decide on a
common naming, even with the ability of chirality.

\begin{figure}[!htbp]
\begin{center}
    \epsfig{file=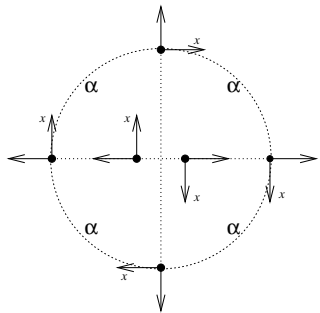, width=0.25\linewidth}
\end{center}
\caption{A symmetric example where six robots are scattered in the plane in such a 
way that for every robot, there is another robot having the same view---the arrows show
the local $x$-$y$ coordinate system of the each robot.  They are not able to agree
on a common direction nor a common naming.} \label{fig:noSoD}
\end{figure}

We now describe our method in the design of a relative (w.r.t. each
robot) naming allowing to implement one-to-one communication for
anonymous robots having no common sense of direction. We refer to
Figure~\ref{fig:SEC} to explain our scheme.

Our method starts (at $t_0$) with the two preprocessing steps
describe above.  At the end we have the Voronoi Diagram and the
sliced granulars---to avoid useless overload in the figure, the
latter are omitted in Figure~\ref{fig:SEC}. Then, still at time
$t_0$, each robot $r$ computes the \emph{smallest enclosing circle},
denoted by $SEC$, of the robot positions. Note that since the robots
have the ability of chirality, then they can agree on a common
clockwise direction of $SEC$.  
Note that $SEC$ is unique and can be computed in linear time~\cite{M83}.

Next, $r$ considers the ``\emph{horizon line}'', denoted by $H_r$,
as the line passing through itself and $\tO$, the center of $SEC$.
Given $H_r$, $r$ consider each radius of $SEC$ passing through a
robot.  The robots are numbered in the growing order following the
radii in the clockwise direction starting from $H_r$.  When several
robots are located on the same radius, they are numbered in the
growing order starting from $\tO$. Note that this means that $r$ is
not necessary labeled by $0$ if some robots are located between
itself and $\tO$ on its radius. An example of this preprocessing
phase is shown in Figure~\ref{fig:SEC} for a given robot $r$.

\begin{figure}[!htbp]
\begin{center}
    \epsfig{file=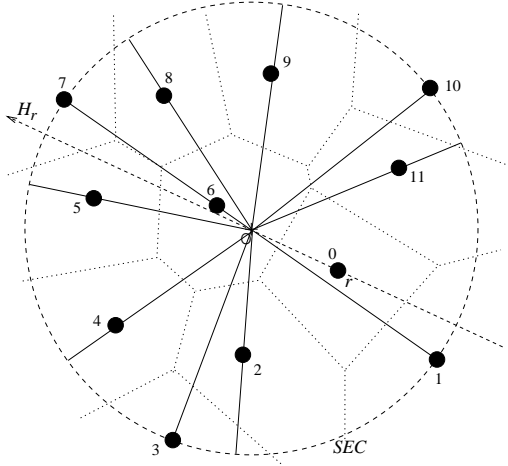, width=0.4\linewidth}
\end{center}
\caption{An example showing how the relative naming is built with respect to Robot $r$.}
\label{fig:SEC}
\end{figure}

The method to send messages to a given robot is the same as the
previous case.  Every robot $r$ slices its granular according to
$H_r$, the diameter corresponding to $H_r$ being labeled by $0$ and
so on in the clockwise direction. The sending of bit is made
following the same scheme as above, the Northern being given by the
direction of $H_r$ and the Eastern following the clockwise
direction. Each robot addresses bits according to its relative
labeling. 

Note first that every robot observes the movements of all the 
robots.  So, every robot is able to know all the messages sent in the system.
This could provide fault-tolerance by redundancy, any robot being
able to send any message again to its addressee. 

Also, note that by construction, the labelling is specific to 
each robot. However, every robot, $r$, is able to compute the labelling
with respect to each robot of the system.  Therefore, for by observing each 
movement made by any robot $r'$, $r$ is able to know to whom a bit is 
addressed, and in particular, to itself.
Every robot is able to compute the message address,
each of them being able to compute the relative naming of all the
robots.

\section{One-to-One Communications in Asynchronous Settings}
\label{sec:async}

In synchronous systems, since every robot is active at each time instant, there is no concern with the receipt of the messages. 
Indeed, every robot motion is seen by every robot.  As a consequence, no message acknowledgment is required. 
By contrast, in asynchronous settings, only fairness is assumed, \ie in each computation steps, some robots can remain inactive.  
So, some robots motions can be lost, and by the way, some messages.  
%
%
Therefore, asynchronous settings require a synchronization mechanism ensuring the acknowledgment of each message sent. 

In this section, we mainly focus on the both {\em Emission} and {\em Receipt} properties.  We first state the two 
following results---due to the lack of space, the proof of Lemma~\ref{lem:1} is given in this appendix:

\begin{lemma}
\label{lem:1}
Let $r$ and $r'$ be two robots. Assume that $r$ always moves in the same direction each time it becomes active. 
If $r$ observes that the position of $r'$ has changed twice, then $r'$  must have observed that the position of $r$ has changed at least once. 
\end{lemma}

\begin{proof}
By contradiction, assume that at time $t_i$, $r$ notes that the position of $r'$ has changed twice and $r'$ has not observed that the position of $r$ has changed at least once. 
Without loss of generality, we assume that $t_i$ is the first time for which $r$ notes that the position of $r'$ has changed twice. So at time $t_i$, $r$ knows three distinct positions of $r'$ and $t_i\geq2$. Let $p_j$ be the last (or the third) position of $r'$ that $r$ has observed, and $t_j$ the first time instant for which $p_j$ is occupied by $r'$. Obviously, $t_j<t_i$. Now we have two cases to consider :
\begin{itemize}
\item {\bf case 1 : $t_j=t_i-1$ }.  The fact that $r$ knows three distinct positions of $r'$ implies that $r$ is become active and has moved at least twice between $t=0$ and $t=t_i-1$ and, thus at least once between $t=0$ and $t=t_i-2$. Consequently, at time $t_j=t_i-1$, $r'$ would have noted that $r$'s position has changed at least once. Contradiction.   
\item {\bf case 2 : $t_j\leq t_i-2$}. We have two cases to consider :
\begin{itemize}
\item {\bf case a : } $r'$ moves at least once between $t_j+1$ and $t_i-1$. In that case, $r$  notes that the position of $r'$ has changed twice before time $t_i$ : that contradicts the fact that $t_i$ is the first time for which $r$ notes that the position of $r'$ has changed twice.   
\item {\bf case b : }$r'$ does not move between $t_j+1$ and $t_i-1$.  As mentioned above, the fact that $r$ knows three distinct positions of $r'$ implies that $r$ is become active and has moved at least twice between $t=0$ and $t=t_i-1$. However, $r'$ does not move between $t_j+1$ and $t_i-1$. Hence, $r$ has moved at least twice between  $t=0$ and $t=t_j$ and, thus at least once between $t=0$ and $t=t_j-1$. So, at time $t_j$ $r'$ would have noted that $r$'s position has changed at least once. Contradiction.  
\end{itemize}  
\end{itemize}     
\end{proof}


\begin{corollary}
\label{cor:1}
Let $r$ and $r'$ be two robots. Assume that $r$ always moves in the same direction on a line $l$ as soon as it becomes active. 
If $r$ observes that the position of $r'$ has changed twice, then $r'$ knows the line $l$ and the direction toward which $r$ moved. 
\end{corollary}

\subsection{One-to-One Communications With Two Asynchronous Robots} 
\label{sub:async-2}

Both robots follow the same scheme, in the sequel referred to as Protocol~$\A{2}$.  Each time a robot, let us say $r$, becomes active, it moves in the opposite direction of 
the other robot, $r'$.  Let us call this direction the $North_r$. 
Robot $r$ behaves like this while it has nothing to send to $r'$.  As soon as $r$ observes that
the position of $r'$ changed twice, by Corollary~\ref{cor:1}, $r$ is guaranteed that $r'$ knows the line $H$ and the direction on which 
$r$ has moved.  Let us call the line $H$ the \emph{horizon line}.  Note that since the two robots follow the same behavior,
$H$ is common to both of them and their respective North are oriented in the opposite direction.  

%

 From this point on, $r$ can start to send messages to $r'$.  
When $r$ wants to send a bit ``$0$'' (``$1$'', respectively) to $r'$, $r$ moves along a line perpendicular to
on the Est side (West side, resp.) of $H$ with respect to $North_r$.  It then move in the same direction each time 
it becomes active until it observes that the position of $r'$ has changed twice.  At this moment on, from 
Lemma~\ref{lem:1}, $r$ knows that $r'$ has seen it in its Est side.  
Then, $r$ comes back on $H$. 
Once $r$ is located on $H$, it starts to move again toward the $North_r$ direction until it observes 
that $r'$ has moved twice. In this way, if Robot~$r$ wants to communicate another bit (following the same scheme), 
it is allowed to move on its East or West side again. So, the new bit and the previous bit are well distinguished 
by Robot~$r'$ even if they have the same value.
An example of our scheme is shown in Figure~\ref{async-2r}.

\begin{figure}[!htbp]
\begin{center}
    \epsfig{file=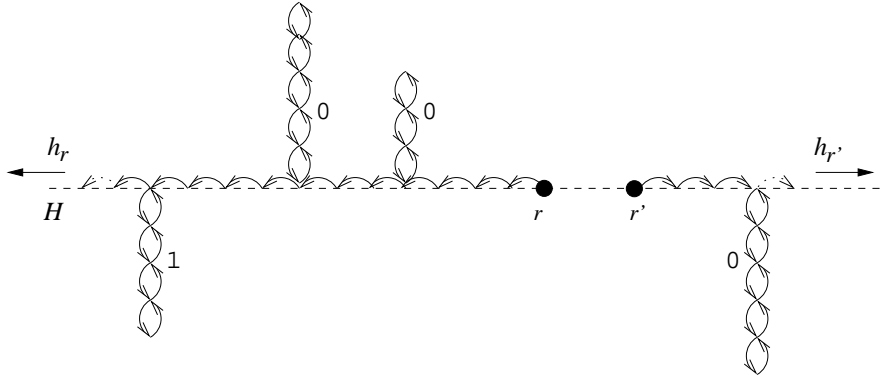, width=0.7\linewidth}
\end{center}
\caption{Asynchronous one-to-one communication for 2 robots. Robot~$r$ sends ``$001\ldots$'', Robot~$r'$ sends ``$0\ldots$''.}
\label{async-2r}
\end{figure}

Note that by Lemma~\ref{lem:1}, Protocol~$\A{2}$ ensures the {\em Receipt} property provided the following condition: 
$r$ observed that the position of $r'$ changed twice before any direction change.  We now show that Protocol~$\A{2}$ ensures this condition. 

\begin{remark}
\label{rem6}
If any robot becomes active, then it moves.
\end{remark}

\begin{lemma}
\label{lem3}
Let $r$ and $r'$ be two robots.  In every execution of Protocol~$\A{2}$, 
$r$ observes that the position of $r'$ changes infinitely often.
\end{lemma}

\begin{proof}
Assume by contradiction that, there exists some executions of Protocol~$\A{2}$ such that, 
eventually, $r$ observes that that position of $r'$ remains unchanged. 
Consider the suffix of such an execution where $r$ observes that the position of $r'$ remains unchanged. 
Assume that $r'$ is eventually motionless.  By fairness and Remark~\ref{rem6}, this case is impossible.
So, $r'$ moves infinitely often.  Thus, each time that $r$ observes $r'$, $r'$ is at the same position. 
There are two cases to consider:
\begin{enumerate}
\item Robot $r'$ eventually sends no bit.  In that case, by executing Protocol~$\A{2}$, $r'$ moves infinitely often in the 
same direction on $H$.  Which contradicts that each time that $r$ observes $r'$, $r'$ at the same position. 

\item Robot $r'$ sends bits infinitely often.  Since $r'$ is at the same position each time $r$ observes it, $r'$ 
goes in a direction and comes back at the same position infinitely often.  So, from Protocol~$\A{2}$, $r'$ 
observed that the position of $r$ changed twice infinitely often (each time $r'$ changes its direction).  
By Lemma~\ref{lem:1}, each time $r'$ observes that the position changed twice, $r$ observes that the position 
of $r'$ has changed at least once.  A contradiction.
\end{enumerate}
\end{proof}

Lemmas~\ref{lem:1} and~\ref{lem3} ensures that Protocol~$\A{2}$ ensures the {\em Receipt} property.
Furthermore, Lemma~\ref{lem3} ensures that no robot is starved sending a bit (\ie it can change
its direction infinitely often).  So, Property {\em Emission} is guaranteed by Protocol~$\A{2}$. 
That leads to the following theorem:

\begin{theorem}
Protocol~$\A{2}$ implements one-to-one explicit communication for two robots. 
\end{theorem}

However, this scheme has the drawback of making the two robots moving away infinitely often of each other.
To deal with this drawback, the two robots can alternate their direction on $H$ each time they send a new bit. 
Lemmas~\ref{lem:1} and~\ref{lem3} ensure that no bit would be missed. 
However, with this protocol, the robots can collide together.  To avoid that the robots neither collide, 
they can divide the covered distance by $x>1$ in each move.  This introduces another drawback: the robots 
are required to be able to move an infinitesimally small distance.  This is also the main drawback of
the protocol of any number of robots in the next subsection.  

\subsection{One-to-One Communication For Any Number Of Anonymous Asynchronous Robots}
\label{sub:async-n}

We now combine the previous method with the one developed in Section~\ref{sec:sync}.
In the following, we describe our method assuming anonymous robots , devoid of sense of direction, \ie 
the weakest assumptions made in Section~\ref{sec:sync}.
Obviously, our method is also compatible with extra assumptions, namely IDs or sense of direction.

We assume that the robots knows $P(t_0)$, \ie the positions of the robots are known
by every robot in $t_0$ or all the robots are awake in $t_0$.
Using $P(t_0)$, when a robot wakes for the time (possibly after $t_0$), it computes all
the preprocessing steps, including the computation of $SEC$.  The only difference is that the granular
is sliced in $n+1$ slices---instead of $n$ slices as in the synchronous case.  Consider that the extra slice,  
corresponding to $H_r$, the diameter being on radius of $SEC$ passing through $r$) is not assigned to a particular
robot.  Let us call this slice $\kappa$.  In our method, $\kappa$ plays the role of the horizon line $H$ as for the 
case with two robots---refer to Figure~\ref{asynchro}.  
That is, each robot moves on its respective $\kappa$ to indicate that it has no bit to transmit.

\begin{figure}[!htbp]
\begin{center}
    \epsfig{file=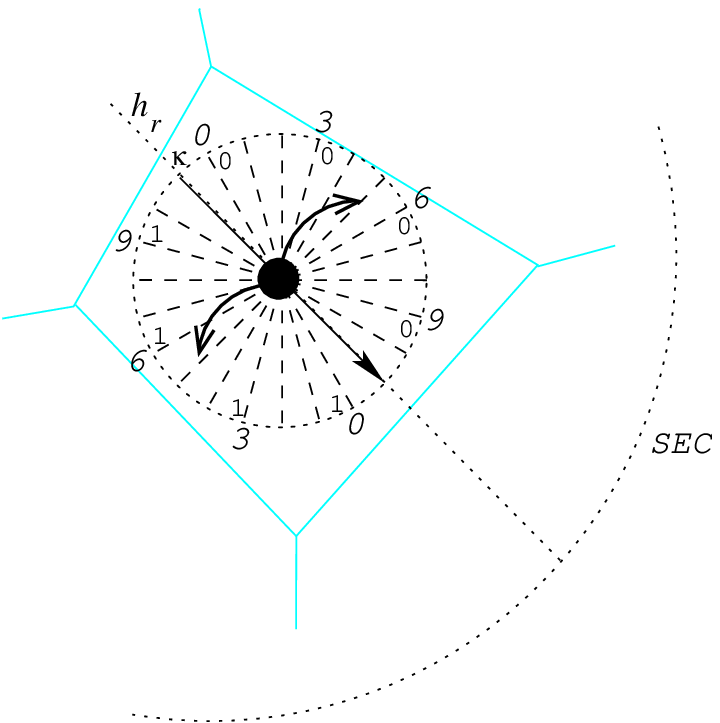, width=0.3\linewidth}
\end{center}
\caption{An example showing how the granular is sliced with respect to $r$.
      \label{asynchro}
   }
\end{figure}

Thus, for every robot $r$, Protocol $\A{n}$ works as follows:
\begin{enumerate}
\item If $r$ wants to send a bit to a robot $r'$, then if $r$ is not at the center of its granular then it comes 
back on this latter one. Then, $r$ moves on the Northern/Estern/North-Estern 
(resp., Southern/Western/South-Western) side on the diameter labelled with $r'$. It then moves in the same direction on the 
diameter labelled $r'$ each time it becomes active.  To make sure that all the robots observed its movement,
$r$ moves in the same direction  until it observes that the position of every robot changed twice. 
Next, $r$ comes back to the center of its granular (\ie 
its initial position in $P(t_0)$ by moving along the same slice. Finally, it moves in the same direction 
on $\kappa$ avoiding to reach the border of its granular before it observes that the position of every 
robot changed twice. In this way, it is now allowed to send a new bit by following the same 
scheme and we are sure that the new bit and the previous bit are well distinguished by 
Robot~$r'$ even if they have the same value.

\item If $r$ does not want to send a bit, then while $r$ does not want to send any bit, it moves in the same direction 
on $\kappa$ avoiding to reach the border of its granular before it observes that the position of every robot changed twice.
When the latter occurs, it moves on the opposite direction, and so on.  If eventually, $r$ wants to send a bit,
 then it comes back toward the center following the same rule as the first case.
\end{enumerate} 

Note that in both case, while a robot is moving toward a direction (on a slice), it must avoid to reach 
either the border or the center of the granular.  In order to implement this constraint, it is required
the robots divide the covered distance by $x>1$ in each move. Also, as for the synchronous case, every 
robot is able to read every message sent by any robot $r$ to any robot $r'$, enabling a kind of redundancy. 

By applying the same reasoning as in Subsection~\ref{sub:async-2} for every pair of robots, we can claim:

\begin{theorem}
Protocol~$\A{n}$ implements one-to-one explicit communication for two robots. 
\end{theorem}

\section{Concluding Remarks, Extensions, and Open Problems}
\label{sec:conclu}

In this work we proposed (deterministic) movement protocols that
implement explicit communication, and therefore
allow the application of (existing) distributed algorithms;
distributing algorithms that use message exchanges.
Movements-signals are introduced as a mean to transfer messages
between deaf and dumb robots.
The movement protocols can serve as a backup to other means of (e.g.,
wireless) communication.

In this paper we first presented protocols for synchronous robots. Some of
them assume
that the robots are equipped with observable IDs
and that the robots agree on a common direction.
We also proposed protocols enabling one-to-one communication among anonymous
robots by building recognition mechanisms assuming the robots
share a common handedness only. Note that the robots may decide to flock in
a certain direction, subtracting
the agreed upon global flocking movement in order to preserve the
relative movements used for communication.

We have adapted the protocols developed in synchronous settings to fit
the general case of asynchronous one-to-one communication among any
number of robots.\\

Next, we list several notable extensions and open problems.

\paragraph{Silent, Finite Movements and Discrete Computation.}
We call the ability for a communication protocol to be {\em silent} when
a robot eventually moves if it has some message to transmit.
Note that this desirable property would help to save energy resources
of the robots.
The protocols proposed with synchronous settings
(Section~\ref{sec:sync}) are clearly silent.
Our asynchronous solutions are not silent (Remark~\ref{rem6}).
The question of whether the design of silent asynchronous algorithms is
possible or
not is open.
A related issue concerns the {\em distance (eventually) covered} by the
robots.
In this paper, we could not overcome either infinite growing distance
or infinitely decreasing distance.  As a matter of fact, we believe
that the fact that
an asynchronous robot that sequentially observes another robot at the
same place, cannot determine
whether the robot moved and returned to the same position or did not
move at all, implies an
impossibility for a communication by a finite number of moves.

Computations with an {\em infinite decimal precision} is different
and, in a way, weaker assumption
than infinitely small movements.  Indeed, one can assume infinite
decimal precision with
the ``reasonable'' assumption of finite movements, \ie with a {\em
minimal} and {\em maximal} distance
covered in one atomic step, or even step over a grid. In this paper,
we assumed a maximal covered distance ($\sigma_r$),
but not a minimal covered distance. This would be the case by assuming
that the plane is
either a grid or a hexagonal pavement~\cite{GWB03}.
For instance, with such assumption, the robots could be prone to make
computation errors due to round off,
and, therefore, face a situation where robots are not able to identify
all of possible $2n$ directions
obtained by slices inside of disks and are limited to recognize only
a certain number of directions. This case could be solved
by avoiding the use of $2n$ slices of granular by transmitting the index of
the robot to whom
the message intended following the message itself. For this we would
need only $k+1$, $1 \leq k < 2n$ segments (or $2k +1$ slices). In
particular, we would use one segment for message transmission (as in
the case of two robots); using the other $k$ segments the robot who
wants to transmit a message allows to transmit the index of the
robot to whom the message is designated. Definitely, such index can
be represented by $\frac{\log{n}}{\log{k}}=\log_{k}n$ symbols.
Notice, this strategy would slow down the algorithm and increase the
number of steps required to transmit a message. More
precisely, the number of steps required in this method to identify
the designated robot is $\log_{k}n$. For example, by taking
$O(\log{n})$ slices instead of $O(n)$, the number of steps to
transmit a message would increase by $O(\frac{\log{n}}{\log\log{n}})$.

\paragraph{Partial synchrony.}
In the continuity of the above discussion, an other important feature
in the field of
mobile robots is the weakness/strength of the model.
For instance, in this paper, we used the semi-synchronous model (SSM).
 It would be interesting to
achieve solutions by {\em relaxing synchrony} among the robots to
achieve solutions into a fully
asynchronous model (\eg, CORDA~\cite{P02}).

\paragraph{Visibility.}
Another issue would the {\em visibility capability} of the
robots~\cite{FPSW99(b)}.
For instance, the following question could be investigated: ``Can
one-to-one communication be achieved by a team of robots with
limited visibility?''  With the same approach, the problem of explicit
communication could be
addressed assuming that the amount of {\em memory} the robots are
equipped with is bounded.

\paragraph{Stabilization.}
{\em Stabilization}~\cite{D00} would be a very desirable property to enable.
It seems that, in our case, stabilization can be achieved in the
synchronous case by carefully adapting
the protocols proposed in Section~\ref{sec:sync}; say by assuming a
global clock (using GPS input)
returning to the initial location and (re)computing the prepossessing
phase every round timestamp.
The self-stabilization property for the asynchronous case requires
further study.

\bibliographystyle{plain}
\bibliography{comm_robots}

\begin{thebibliography}{10}

\bibitem{A91}
F~Aurenhammer.
\newblock Voronoi diagrams- a survey of a fundamental geometric data structure.
\newblock {\em ACM Comput. Surv.}, 23(3):345--405, 1991.

\bibitem{BHD94}
R.~Beckers, O.E. Holland, and J.L. Deneubourg.
\newblock From local actions to global tasks: Stigmergy and collective
  robotics.
\newblock {\em Artificial Life, 4th Int. Worksh. on the Synth. and Simul. of
  Living Sys.}, 4:173--202, 1994.

\bibitem{CFPS03}
M~Cieliebak, P~Flocchini, G~Prencipe, and N~Santoro.
\newblock Solving the robots gathering problem.
\newblock In {\em Proceedings of the 30th International Colloquium on Automata,
  Languages and Programming (ICALP 2003)}, pages 1181--1196, 2003.

\bibitem{CP08}
R~Cohen and D~Peleg.
\newblock Local spreading algorithms for autonomous robot systems.
\newblock {\em Theor. Comput. Sci.}, 399(1-2):71--82, 2008.

\bibitem{K97}
C.R.Kube.
\newblock Task modelling in collective robotics.
\newblock {\em Auton. Robots}, 4(1):53--72, 1997.

\bibitem{DK02}
X~Defago and A~Konagaya.
\newblock Circle formation for oblivious anonymous mobile robots with no common
  sense of orientation.
\newblock In {\em 2nd ACM International Annual Workshop on Principles of Mobile
  Computing (POMC 2002)}, pages 97--104, 2002.

\bibitem{DLP08}
Y~Dieudonn{\'e}, O~Labbani-Igbida, and F~Petit.
\newblock Circle formation of weak mobile robots.
\newblock {\em ACM Transactions on Autonomous and Adaptive Systems}, 3(4),
  2008.

\bibitem{DP07c}
Y~Dieudonn{\'e} and F~Petit.
\newblock Deterministic leader election in anonymous sensor networks without
  common coodinated system.
\newblock In {\em 11th International Conference On Principles of Distributed
  Systems (OPODIS 2007)}, volume 4878 of {\em Lecture Notes in Computer
  Science, Springer}, pages 132--142, 2007.

\bibitem{D00}
S.~Dolev.
\newblock {\em Self-Stabilization}.
\newblock The MIT Press, 2000.

\bibitem{DGS96}
S.~Dolev, MG~Gouda, and M~Schneider.
\newblock Memory requirements for silent stabilization.
\newblock In {\em PODC96 Proceedings of the 15th Annual ACM Symposium on
  Principles of Distributed Computing}, pages 27--34, 1996.

\bibitem{W75}
E.O.Wilson.
\newblock {\em Sociobiology}.
\newblock Belknap Press of Harward University Press, 1975.

\bibitem{FPSW99}
P~Flocchini, G~Prencipe, N~Santoro, and P~Widmayer.
\newblock Hard tasks for weak robots: The role of common knowledge in pattern
  formation by autonomous mobile robots.
\newblock In {\em 10th Annual International Symposium on Algorithms and
  Computation (ISAAC 99)}, pages 93--102, 1999.

\bibitem{FPSW99(b)}
P~Flocchini, G~Prencipe, N~Santoro, and P~Widmayer.
\newblock Gathering of autonomous mobile robots with limited visibility.
\newblock In {\em STACS 2001}, pages 247--258, 2001.

\bibitem{FPSW01}
P~Flocchini, G~Prencipe, N~Santoro, and P~Widmayer.
\newblock Pattern formation by autonomous robots without chirality.
\newblock In {\em VIII International Colloquium on Structural Information and
  Communication Complexity (SIROCCO 2001)}, pages 147--162, 2001.

\bibitem{FN88}
T.~Fukuda and S.~Nakagawa.
\newblock Approach to the dynamically reconfigurable robotic sytem.
\newblock {\em Journal of Intelligent and Robotic System}, 1:55--72, 1988.

\bibitem{GWB03}
N~Gordon, I~A. Wagner., and A~M. Brucks.
\newblock Discrete bee dance algorithms for pattern formation on a grid.
\newblock In {\em IAT '03: Proceedings of the IEEE/WIC International Conference
  on Intelligent Agent Technology}, pages 545--549. IEEE Computer Society,
  2003.

\bibitem{G59}
P-P Grass{\'e}.
\newblock La reconstruction du nid et les coordinations inter-individuelles
  chez bellicosi-termes natalensis et cubitermes sp. la theorie de la
  stigmergie: Essai d’interpretation des termites constructeurs.
\newblock {\em Insectes Sociaux}, 6:41--83, 1959.

\bibitem{M95}
M.J. Matari{\'c}.
\newblock Issues and approaches in the design of collective autonomous agents.
\newblock {\em Robotics and Autonomous Systems}, 16(2-4):321--331, 1995.

\bibitem{M83}
N~Megiddo.
\newblock Linear-time algorithms for linear programming in ${\mathbb r}^3$ and
  related problems.
\newblock {\em SIAM Journal on Computing}, 12(4):759--776, 1983.

\bibitem{N92}
F.R. Noreils.
\newblock An architecture for cooperative and autonomous mobile robots.
\newblock In {\em IEEE International Conference on Robotics and Automation},
  pages 2703--2710, 1992.

\bibitem{P02}
G~Prencipe.
\newblock {\em Distributed Coordination of a Set of Autonomous Mobile Robots.}
\newblock PhD thesis, Dipartimento di Informatica, University of Pisa, 2002.

\bibitem{PS06}
G~Prencipe and N~Santoro.
\newblock Distributed algorithms for autonomous mobile robots.
\newblock In {\em The 2006 IFIP International Conference on Embedded And
  Ubiquitous Computing (EUC 2006)}, 2006.

\bibitem{SY96}
I~Suzuki and M~Yamashita.
\newblock Agreement on a common $x$-$y$ coordinate system by a group of mobile
  robots.
\newblock {\em Intelligent Robots: Sensing, Modeling and Planning}, pages
  305--321, 1996.

\bibitem{SY99}
I~Suzuki and M~Yamashita.
\newblock Distributed anonymous mobile robots - formation of geometric
  patterns.
\newblock {\em SIAM Journal of Computing}, 28(4):1347--1363, 1999.

\bibitem{WW90}
P.H Wenner and A.M Wells.
\newblock {\em Anatomy of a Controversy: the Question of a ``Language'' Among
  Bees}.
\newblock Columbia University Press, New-York, 1990.

\end{thebibliography}
\end{document}